\documentclass[10pt,twocolumn]{ieeeconf}
\usepackage{color}
\usepackage{float}
\usepackage{multirow}
\usepackage{amsmath}
\usepackage{amssymb,eqnarray}
\usepackage{mathabx}
\usepackage{graphicx}
\usepackage{stmaryrd}
\usepackage{psfrag}
\usepackage{pstool}

\usepackage{tikz}
\usetikzlibrary{automata, positioning, arrows}
\tikzset{
->, 
>=latex,
node distance=3.5cm, 
every state/.style={thick, fill=blue!15}, 
initial text=$ $, 
}

\DeclareMathOperator*{\argmin}{arg\,min}

\makeatletter
\newcommand{\pushright}[1]{\ifmeasuring@#1\else\omit\hfill$\equationstyle#1$\fi\ignorespaces}
\newcommand{\pushleft}[1]{\ifmeasuring@#1\else\omit$\equationstyle#1$\hfill\fi\ignorespaces}
\makeatother

\DeclareMathOperator*{\esssup}{ess\,sup}

\usepackage{multicol}
\usepackage{booktabs} 
\usepackage[noend]{algorithmic}
\usepackage{algorithm2e}
\usepackage[english]{babel}
\usepackage{subfigure}
\usepackage{mathtools}
\mathtoolsset{showonlyrefs=false}
\makeatletter
\usepackage{multirow}   
\usepackage{array, makecell} %

\usepackage{multicol}
\usepackage[noend]{algorithmic}
\usepackage{algorithm2e}

\newtheorem{theorem}{\protect\theoremname}
\newtheorem{defn}{\protect\definitionname}
\newtheorem{proposition}{\protect\propositionname}

\newtheorem{lemma}{\protect\lemmaname}

\newtheorem{assum}{\protect\assumname}
\newtheorem{problem}{\protect\probname}

\providecommand{\definitionname}{\textbf{Definition}}
\providecommand{\propositionname}{\textbf{Proposition}}
\providecommand{\remarkname}{\textbf{Remark}}
\providecommand{\theoremname}{\textbf{Theorem}}
\providecommand{\lemmaname}{\textbf{Lemma}}
\providecommand{\assumname}{\textbf{Assumption}}
\providecommand{\probname}{\textbf{Problem}}
\providecommand{\corname}{\textbf{Corollary}}

\IEEEoverridecommandlockouts

\usepackage[margin=.7in]{geometry}
\title{Risk-Averse Stochastic Shortest Path Planning}
\author{Mohamadreza Ahmadi, Anushri Dixit, Joel W. Burdick, and Aaron D. Ames \thanks{The authors are with the California Institute of Technology, 1200 E. California Blvd., MC 104-44, Pasadena, CA 91125,  e-mail: (\{mrahmadi, adixit, ames\}@caltech.edu, jwb@robotics.caltech.edu.) }}

\begin{document}

\maketitle

\begin{abstract}
We consider the stochastic shortest path planning problem in MDPs, \textit{i.e.,} the problem of designing policies that ensure reaching a goal state from a given initial state with minimum accrued cost. In order to account for rare but important realizations of the system, we consider a nested dynamic coherent risk total cost functional rather than the conventional risk-neutral total expected cost. Under some assumptions, we show that optimal, stationary, Markovian policies exist and can be found via a special Bellman's equation. We propose a computational technique based on difference convex programs (DCPs) to find the associated value functions and therefore the risk-averse policies. A rover navigation MDP is used to illustrate the proposed methodology with conditional-value-at-risk (CVaR) and entropic-value-at-risk (EVaR) coherent risk measures.
\end{abstract}

\section{Introduction}

Shortest path problems~\cite{bellman1958routing}, \textit{i.e.,} the problem of reaching a goal state form an initial state with minimum total cost, arise in several real-world applications, such as driving directions on web mapping websites like MapQuest or Google Maps~\cite{googletalk} and robotic path planning~\cite{chen1996developing}. In a shortest path problem, if transitions from one system state to another is subject to stochastic uncertainty, the problem is referred to as a stochastic shortest path (SSP) problem~\cite{bertsekas1991analysis,sigal1980stochastic}. In this case, we are interested in designing policies such that the \textit{total expected} cost is minimized.  Such planning under uncertainty problems are indeed equivalent to an undiscounted total cost Markov decision processes (MDPs)~\cite{Puterman94} and can be solved efficiently via the dynamic programming method~\cite{bertsekas1991analysis,bertsekas2013stochastic}. 

 {However, emerging applications in path planning, such as autonomous navigation in extreme environments, \textit{e.g.}, subterranean~\cite{fan2021step} and extraterrestrial environments~\cite{ahmadi2020risk}, not only require reaching a goal region, but also risk-awareness for mission success.  Nonetheless, the conventional total expected cost is only meaningful if the law of large numbers can be invoked and it ignores important but rare system realizations. In addition, robust planning solutions may give rise to behavior that is extremely conservative.}


Risk can be quantified in numerous ways. For example, mission risks can be mathematically characterized in terms of chance constraints~\cite{ono2013probabilistic,ono2015chance}, utility functions~\cite{dvijotham2014convex}, and distributional robustness~\cite{xu2010distributionally}. {Chance constraints often account for Boolean events (such as collision with an obstacle or reaching a goal set) and do not take into consideration the tail of the cost distribution. To account for the latter,  risk measures have been advocated for planning and decision making tasks in robotic systems~\cite{majumdar2020should}.} The preference of one risk measure over
another depends on factors such as sensitivity to rare events, ease of estimation from data, and computational tractability. Artzner \textit{et. al.}~\cite{artzner1999coherent} characterized a set of natural properties that are desirable for a risk measure, called a coherent risk measure, and  have henceforth obtained widespread
acceptance in finance and operations research, among others. Coherent risk measures can be interpreted as a special form of distributional robustness, which will be leveraged later in this paper. 

Conditional value-at-risk (CVaR) is an important coherent risk measure that has received significant attention in decision making problems, such as MDPs~\cite{chow2015risk,chow2014algorithms,prashanth2014policy,bauerle2011markov}. General coherent risk measures for MDPs were studied in~\cite{ruszczynski2010risk,ahmadi2020constrained}, wherein it was further assumed the risk measure is \emph{time consistent}, akin to the dynamic programming property. Following the footsteps of~\cite{ruszczynski2010risk}, \cite{tamar2016sequential} proposed a sampling-based algorithm for MDPs with static and dynamic coherent risk measures using policy gradient and actor-critic methods, respectively (also, see a model predictive control technique for linear dynamical systems with coherent risk objectives~\cite{singh2018framework}). A method based on stochastic reachability analysis was proposed in~\cite{chapman2019risk} to estimate a CVaR-safe set of initial conditions via the solution to an MDP. A worst-case CVaR SSP planning method was proposed and solved via dynamic programming in~\cite{gavriel2012risk}. Also, total cost undiscounted MDPs with static CVaR measures were studied in~\cite{carpin2016risk} and solved via a surrogate MDP, whose solution approximates the optimal policy with arbitrary accuracy. 


In this paper, we propose a method for designing policies for SSP planning problems, such that the total accrued cost in terms of dynamic, coherent risk measures is minimized (a generalization of the problems considered in \cite{gavriel2012risk} and~\cite{carpin2016risk} to dynamic, coherent risk measures). We begin by showing that, under the assumption that the goal region is reachable in finite time with non-zero probability, the total accumulated risk cost is always bounded. We further show that, if the coherent risk measures satisfy a Markovian property, we can find optimal, stationary, Markovian risk-averse policies via solving a special Bellman's equation. We also propose a computational method based on difference convex programming to solve the Bellman's equation and therefore design risk-averse policies. We elucidate the proposed method via numerical examples involving a rover navigation MDP and CVaR and entropic-value-at-risk (EVaR) measures.

The rest of the paper is organized as follows. In the next section, we review some definitions and properties used in the sequel. In Section III, we present the problem under study and show its well-posedness  under an assumption. In Section IV, we present the main result of the paper, \textit{i.e.,} a special Bellman's equation for the risk-averse SSP problem. In Section V, we describe a computational method to find risk-averse policies. In Section VI, we illustrate the proposed method via a numerical example and finally, in Section VI, we conclude the paper.

\textbf{Notation: } We denote by $\mathbb{R}^n$ the $n$-dimensional Euclidean space and $\mathbb{N}_{\ge0}$ the set of non-negative integers. We use bold font to denote a vector and $(\cdot)^\top$ for its transpose, \textit{e.g.,} $\boldsymbol{a}=(a_1,\ldots,a_n)^\top$, with $n\in \{1,2,\ldots\}$. For a vector $\boldsymbol{a}$, we use $\boldsymbol{a}\succeq (\preceq) \boldsymbol{0}$ to denote element-wise non-negativity (non-positivity) and $\boldsymbol{a}\equiv \boldsymbol{0}$ to show all elements of $\boldsymbol{a}$ are zero. For a finite set $\mathcal{A}$, we denote its power set by $2^\mathcal{A}$, \textit{i.e.,} the set of all subsets of $\mathcal{A}$. For  a probability space $(\Omega, \mathcal{F}, \mathbb{P})$ and a constant $p \in [1,\infty)$, $\mathcal{L}_p(\Omega, \mathcal{F}, \mathbb{P})$ denotes the vector space of real valued random variables $c$ for which $\mathbb{E}|c|^p < \infty$. Superscripts are used to denote indices and subscripts are used to denote time steps (stages), \textit{e.g.}, for $s \in \mathcal{S}$, $s_1^2$ means the the value of $s^2 \in \mathcal{S}$ at the $1$st stage.

\section{Preliminaries}

This section, briefly reviews notions and definitions used throughout the paper.




We are interested in designing policies for a class of finite MDPs (termed \textit{transient MDPs} in~\cite{carpin2016risk}) as shown in Figure~\ref{fig:fig1}, which is defined next.

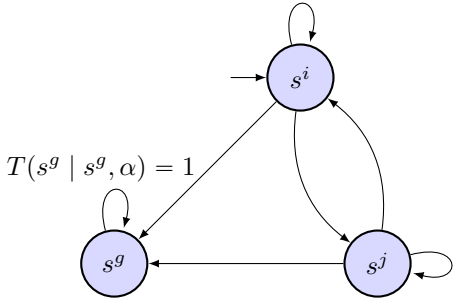
\begin{figure}[t] 
\centering
\begin{tikzpicture}
\node[state, initial] (1) {$s^i$};
\node[state, below left of=1] (2) {$s^g$};
\node[state, right of=2] (3) {$s^j$};
\draw (1) edge[above] node{} (2)
(1) edge[below, bend right, left=0.3] node{} (3)
(1) edge[loop above] node{} (1)
(2) edge[loop above] node{$T(s^g\mid s^g,\alpha)=1 \quad$} (2)
(3) edge[loop right] node{} (3)
(3) edge[below] node{} (2)
(3) edge[above, bend right, right=0.3] node{} (1);
\end{tikzpicture}
\caption{The transition graph of the particular class of MDPs studied in this paper. The goal state $s^g$ is cost-free and absorbing. }
\label{fig:fig1}
\end{figure}

\begin{defn}[MDP] \label{defn:MDP} {
An \emph{MDP} is a tuple, $\mathcal{M}=(\mathcal{S},Act, T, s_0,c,s^g)$, where
\begin{itemize}
\item States $\mathcal{S}  =
\{s^{1} ,\dots,s^{|\mathcal{S}|} \}$ of the
autonomous agent(s) and world model,
\item Actions $Act = \{\alpha^{1},\dots,\alpha^{|Act|}\}$ available to the robot,
\item A transition probability distribution $T(s^{j} |s^{i} ,\alpha)$, satisfying  $\sum_{s  \in \mathcal{S} } T(s |s^{i} ,\alpha) = 1,  \forall s^i  \in \mathcal{S} ,\forall\alpha \in Act$,
\item An initial state $s_0 \in \mathcal{S}$,
and
\item An immediate cost function, $c(s^{i},\alpha^i ) \ge 0$, for each state $s^{i}  \in \mathcal{S} $ and action~$\alpha^i \in {Act}$,
\item $s^g \in \mathcal{S}$ is a special cost-free goal (termination) state, \textit{i.e.,} $T(s^g\mid s^g,\alpha)=1$ and $c(s^g,\alpha)=0$ for all $\alpha \in Act$.
\end{itemize}
}
\end{defn}

We assume the immediate cost function $c$ is non-negative and upper-bounded by a positive constant $\bar{c}$.

Our risk-averse policies for the SSP problem rely on the notion of dynamic coherent risk measures, whose definitions and properties are presented next.

Consider a probability space $(\Omega, \mathcal{F}, \mathbb{P})$, a filtration $\mathcal{F}_0 \subset \cdots \mathcal{F}_T \subset \mathcal{F} $, and an adapted sequence of random variables~(stage-wise costs) $c_t,~t=0,\ldots, T$, where $T \in \mathbb{N}_{\ge 0} \cup \{\infty\}$.
For $t=0,\ldots,T$, we further define the spaces $\mathcal{C}_t = \mathcal{L}_p(\Omega, \mathcal{F}_t, \mathbb{P})$, $p \in [1,\infty)$,  $\mathcal{C}_{t:T}=\mathcal{C}_t\times \cdots \times \mathcal{C}_T$ and $\mathcal{C}=\mathcal{C}_0\times \mathcal{C}_1 \times \cdots$. In order to describe how one can evaluate the risk of sub-sequence $c_t,\ldots, c_T$ from the perspective of stage $t$, we require the following definitions.


\vspace{0.1cm}
\begin{defn}[Conditional Risk Measure]{
A mapping $\rho_{t:T}: \mathcal{C}_{t:T} \to \mathcal{C}_{t}$, where $0\le t\le N$, is called a \emph{conditional risk measure}, if it has the following monotonicity property:
\begin{equation*}
    \rho_{t:T}(\boldsymbol{c}) \le   \rho_{t:T}(\boldsymbol{c}'), \quad \forall \boldsymbol{c}, \forall \boldsymbol{c}' \in \mathcal{C}_{t:T}~\text{such that}~\boldsymbol{c} \preceq \boldsymbol{c}'.
\end{equation*}
}
\end{defn}
\vspace{0.1cm}
\begin{defn}[Dynamic Risk Measure]
{A \emph{dynamic risk measure} is a sequence of conditional risk measures $\rho_{t:T}:\mathcal{C}_{t:T}\to \mathcal{C}_{t}$, $t=0,\ldots,T$.}
\end{defn}
\vspace{0.1cm}
 One fundamental property of dynamic risk measures is their consistency over time~\cite[Definition 3]{ruszczynski2010risk}.
  If a risk measure is time-consistent, we can define the one-step conditional risk measure $\rho_t:\mathcal{C}_{t+1}\to \mathcal{C}_t$, $t=0,\ldots,T-1$ as follows:
\begin{equation}
    \rho_t(c_{t+1}) = \rho_{t,t+1}(0,c_{t+1}),
\end{equation}
and for all $t=1,\ldots,T$, we obtain:
\begin{multline}
    \label{eq:dynriskmeasure}
    \rho_{t,T}(c_t,\ldots,c_T)= \rho_t \big(c_t + \rho_{t+1} ( c_{t+1}+\rho_{t+2}(c_{t+2}+\cdots\\
    +\rho_{T-1}\left(c_{T-1}+\rho_{T}(c_T) \right) \cdots )) \big).
\end{multline}
Note that the time-consistent risk measure is completely defined by one-step conditional risk measures $\rho_t$, $t=0,\ldots,T-1$ and, in particular, for $t=0$, \eqref{eq:dynriskmeasure} defines a risk measure of the entire sequence $\boldsymbol{c} \in \mathcal{C}_{0:T}$.

At this point, we are ready to define a coherent risk measure. 
\begin{defn}[Coherent Risk Measure]\label{defi:coherent}{
We call the one-step conditional risk measures $\rho_t: \mathcal{C}_{t+1}\to \mathcal{C}_t$, $t=1,\ldots,N-1$ as in~\eqref{eq:dynriskmeasure} a \emph{coherent risk measure}, if it satisfies the following conditions
\begin{itemize}
    \item \textbf{Convexity:} $\rho_t(\lambda c + (1-\lambda)c') \le \lambda \rho_t(c)+(1-\lambda)\rho_t(c')$, for all $\lambda \in (0,1)$ and all $c,c' \in \mathcal{C}_{t+1}$;
    \item \textbf{Monotonicity:} If $c\le c'$, then $\rho_t(c) \le \rho_t(c')$ for all $c,c' \in \mathcal{C}_{t+1}$;
    \item \textbf{Translational Invariance:} $\rho_t(c'+c)=\rho_t(c')+c$ for all $c \in \mathcal{C}_t$ and $c' \in \mathcal{C}_{t+1}$;
    \item \textbf{Positive Homogeneity:} $\rho_t(\beta c)= \beta \rho_t(c)$ for all $c \in \mathcal{C}_{t+1}$ and $\beta \ge 0$.
\end{itemize}
}
\end{defn}
\vspace{0.1cm}

In fact, we can show that there exists a dual (or distributionally robust) representation for any coherent risk measure. Let $m,n \in [1,\infty)$ such that $1/m+1/n=1$ and 
 $$
\mathcal{P} = \big\{q \in \mathcal{L}_n(\mathcal{S}, 2^\mathcal{S}, \mathbb{P}) \mid \sum_{s' \in \mathcal{S}} q(s') \mathbb{P}(s')=1,~q\ge 0 \big\}.
$$

\begin{proposition}[Proposition 4.14 in~\cite{follmer2011stochastic}]
Let $\mathcal{Q}$ be a closed convex subset of $\mathcal{P}$. The one-step conditional risk measure $\rho_t:\mathcal{C}_{t+1}\to\mathcal{C}_t$, $t=1,\ldots,N-1$ is a coherent risk measure if and only if
\begin{equation}\label{eq:dual}
    \rho_t(c) = \sup_{q \in \mathcal{Q}}~~ \langle c,q \rangle_{\mathcal{Q}},\quad \forall c \in \mathcal{C}_{t+1},
\end{equation}
where $\langle \cdot, \cdot \rangle_{\mathcal{Q}}$ denotes the inner product in ${\mathcal{Q}}$. 
\end{proposition}
\vspace{0.1cm}



Hereafter, all risk measures are assumed to be coherent.




\section{Problem formulation}

Next, we formally describe the risk-averse SSP problem. We also demonstrate that, if the goal state is reachable in finite time, the risk-averse SSP problem is well-posed. Let $\pi = \{\pi_0,\pi_1,\ldots\}$ be an admissible policy.

\vspace{0.1cm}
\begin{problem}
\textit{Consider MDP $\mathcal{M}$ as described in Definition 1. Given an initial state $s_0 \neq s^g$, we are interested in solving the following problem
\begin{align}\label{eq:problem}
\pi^* \in ~\argmin_{\pi} ~~J(s_0,\pi),
\end{align}
where  
\begin{equation} \label{cveerfdd}
J(s_0,\pi) = \lim_{T \to \infty} \rho_{t:T}\left(c(s_0,\pi_0),\ldots,c(s_T,\pi_T)\right),
\end{equation}
is the total risk functional for the admissible policy $\pi$.
}
\end{problem}
\vspace{0.1cm}
In fact, we are interested in reaching the goal state $s^g$ such that the total risk cost is minimized\footnote{An important class of SSP planning problems are concerned with minimum-time reachability. Indeed, our formulation also encapsulates minimum-time problems, in which for MDP $\mathcal{M}$, we have $c(s)=1$, for all $s \in \mathcal{S}\setminus \{s^g\}$. }. Note that the risk-averse deterministic shortest problem can be obtained as a special case when the transitions are deterministic. We define the optimal risk value function as 
$$
J^*(s)=\min_\pi~J(s,\pi),~~\forall s \in \mathcal{S},
$$
and call a stationary policy $\pi = \{\mu,\mu,\ldots\}$ (denoted $\mu$) optimal if $J(s,\mu)=J^*(s)=\min_\pi~J(s,\pi),~~\forall s \in \mathcal{S}$.

We posit the following assumption, which implies that the goal state is reachable eventually under all policies. 
\vspace{0.1cm}
\begin{assum}[Goal is Reachable in Finite Time] \textit{Regardless of the policy used and the initial state, there exists an integer $\tau$ such that there is a positive probability that the goal state $s^g$ is visited after no more than $\tau$ stages\footnote{If instead of one goal state $s^g$, we were interested in a set of goal states $\mathcal{G} \subset \mathcal{S}$, it suffices to define $\tau = \inf \{t \mid \mathbb{P}(s_t \in \mathcal{G}\mid s_0\in\mathcal{S}\setminus \mathcal{G},\pi)>0\}$. Then, all the paper's derivations can be applied. For the sake of simplicity of the presentation, we present the results for a single goal state.}.}
\end{assum}
\vspace{0.1cm}

We then have the following observation with respect to Problem 1\footnote{Note that Problem 1 is ill-posed in general. For example, if the induced Markov chain for an admissible policy is periodic, then the limit in~\eqref{cveerfdd} may not exist. This is in contrast to risk-averse discounted infinite-horizon MDPs~\cite{ahmadi2020constrained}, for which we only require non-negativity and boundedness of immediate costs.}.
\vspace{0.1cm}
\begin{proposition}
\textit{Let Assumption~1 hold. Then, the risk-averse SSP problem, \textit{i.e.,} Problem 1, is well-posed and $J(s_0,\pi)$ is bounded for all policies $\pi$.}
\end{proposition}
\vspace{0.1cm}
\begin{proof}
Assumption 1 implies that for each admissible policy $\pi$, we have
$
p_\pi =\max_{s\in \mathcal{S}} \mathbb{P}(s_\tau \neq s^g \mid s_0=s,\pi) < 1.
$ 
That is, given a policy $\pi$, the probability $p_\pi$ of not visiting the goal state $s^g$ is less than one. Let $p = \max_{\pi} p_\pi$. Remark that $p_\pi$ is dependent solely on $\{\pi_1,\pi_2,\ldots,\pi_{\tau}\}$. Moreover, since $Act$ is finite, the number of $\tau$-stage policies is also finite, which implies finiteness of $p_\pi$. Hence, $p<1$ as well. Therefore, for any policy $\pi$ and initial state $s$, we obtain 
$
 \mathbb{P}(s_{2\tau} \neq s^g \mid s_0=s,\pi) =   \mathbb{P}(s_{2\tau} \neq s^g \mid s_\tau \neq s^g, s_0=s,\pi)  \times  \mathbb{P}(s_\tau \neq s^g \mid s_0=s,\pi) \le p^2.
$ 
Then, by induction, we can show that, for any admissible SSP policy $\pi$, we have
\begin{equation} \label{fcdfdscsdfc}
 \mathbb{P}(s_{k\tau} \neq s^g \mid s_0=s,\pi) \le p^k,\quad \forall s \in \mathcal{S},
\end{equation}
and $k=1,2,\ldots$. 
Indeed, we can show that the risk-averse cost incurred in the $\tau$ periods between $\tau k$ and $\tau (k+1)-1$ is bounded as follows
\begin{subequations} \label{eq:idk1}
\begin{align}
    ~& \rho_{0}  \big(  \cdots
    \rho_{\tau (k+1)-1}(c_{\tau k}+\cdots+c_{\tau (k+1)-1})  \cdots  \big)  \\ &= \tilde{\rho}(c_{\tau k}+c_{\tau k+1}+\cdots+c_{\tau (k+1)-1}) \label{dwe1} \\  &\le \tilde{\rho}(\bar{c}+\cdots+\bar{c}) \label{dwef2} \\ &=  \sup_{q \in \mathcal{Q}}~\langle \tau \bar{c}, q \rangle_{\mathcal{Q}} \label{dwef3} \\ 
    &\le \langle \tau \bar{c}, q^* \rangle_{\mathcal{Q}}  \label{dwef4} \\
    &= \tau \bar{c}  \sum_{s \in \mathcal{S}} \mathbb{P}(s_{k\tau} \neq s^g \mid s_0=s,\pi) q^*(s,\pi) \\
    &\le \tau \bar{c}  \times \sup \left( \mathbb{P}(s_{k\tau} \neq s^g \mid s_0=s,\pi) \right)  \sum_{s\in \mathcal{S}} |q^*(s,\pi) | \label{dwef5} \\ &
    \le \tau~\bar{c}~p^k,
\end{align}
\end{subequations}
where in~\eqref{dwe1} we used the translational invariance property of coherent risk measures and defined $\tilde{\rho} = \rho_{0} \circ \cdots \circ \rho_{\tau (k+1)-1} $. Since any finite compositions of the coherent risk measures is a risk measure~\cite{shapiro2014lectures}, we have that $\tilde{\rho}$ is also a coherent risk measure. Moreover, since the immediate cost function $c$ is upper-bounded, from the monotonicity property of the coherent risk measure $\tilde{\rho}$, we obtain~\eqref{dwef2}. Equality~\eqref{dwef3} is derived from Proposition 1. Inequality~(7e) is obtained via defining $q^* = \mathrm{argsup}_{q \in \mathcal{Q}} \langle \tau \bar{c}, q \rangle_{\mathcal{Q}}$. In inequality~\eqref{dwef5}, we used H\"older inequality and finally we used~\eqref{fcdfdscsdfc} to obtain the last inequality. Thus, the  risk-averse total cost $J(s,\pi)$, $s \in \mathcal{S}$, exists and is finite, because given Assumption 1 we have
\begin{equation} \label{eqsdsdsds}
\begin{aligned}
    |J(s_0,\pi)| &= \lim_{T \to\infty}  \rho_{0}  \circ  \cdots \circ \rho_{\tau-1} \circ \cdots \circ \rho_T(c_{0}+\cdots+c_T)  
    \\ \le &\sum_{k=0}^\infty \rho_{0}  \big(  \cdots
    \rho_{\tau (k+1)-1}(c_{\tau k}+\cdots+c_{\tau (k+1)-1})  \cdots  \big) \\
    \le &\sum_{k=0}^\infty \tau~\bar{c}~p^k = \frac{\tau \bar{c}}{1-p},
    \end{aligned}
    \end{equation}
    where in the first equality above we used the  translational invariance property  and in the first inequality we used the sub-additivity property of coherent risk measures. Hence, $J(s_0,\pi)$ is bounded for all $\pi$.
\end{proof}
\section{Risk-Averse SSP Planning}

This section presents the  paper's main result, which includes a special Bellman's equation for finding the risk value functions for Problem 1. Furthermore, assuming that the coherent risk measures satisfy a Markovian property, we show that the optimal risk-averse policies are stationary and Markovian.

To begin with, note that at any time $t$, the value of $\rho_t$ is $\mathcal{F}_t$-measurable and is allowed to depend on the entire history of the process $\{s_0,s_1,\ldots\}$ and we cannot expect to obtain a Markov optimal policy~\cite{ott2010markov}. In order to obtain Markov optimal policies for Problem 1, we need the following property~\cite[Section 4]{ruszczynski2010risk} of risk measures.
\vspace{0.1cm}
\begin{defn}[Markov Risk Measure~\cite{fan2018process}]\label{assum1}\textit{
A one-step conditional risk measure $\rho_t:\mathcal{C}_{t+1}\to \mathcal{C}_t$ is a Markov risk measure with respect to MDP~$\mathcal{P}$, if there exist a risk transition mapping $\sigma_t: \mathcal{L}_m(\mathcal{S}, 2^\mathcal{S}, \mathbb{P}) \times \mathcal{S} \times \mathcal{M} \to \mathbb{R}$ such that for all $v \in \mathcal{L}_m(\mathcal{S}, 2^\mathcal{S}, \mathbb{P})$ and $\alpha_t \in \pi(s_t)$, we have
\begin{equation}
    \rho_t(v(s_{t+1})) = \sigma_t\left(v(s_{t+1}),s_t,T(s_{t+1}|s_t,\alpha_t)\right).
\end{equation}
}
\end{defn}
\vspace{0.1cm}

In fact, if $\rho_t$ is a coherent risk measure, $\sigma_t$ also satisfies the properties of a coherent risk measure (Definition 4). 

\vspace{0.1cm}
\begin{assum}\label{assum1}\textit{
The one-step coherent risk measure $\rho_t$ is a Markov risk measure.
}
\end{assum}
\vspace{0.1cm}

We can now present the main result in the paper, a form of Bellman's equations for solving the risk-averse SSP problem.
\vspace{0.05cm}
\begin{theorem}
\textit{Consider MDP $\mathcal{P}$ as described in Definition~1 and let Assumptions~1 and 2 hold. Then, the following statements are true for the risk-averse SSP problem:\\
\textbf{(i)} Given (non-negative) initial condition $J^0(s), s \in \mathcal{S}$,  the sequence generated by the recursive formula (dynamic programming)
    \begin{multline}\label{eq:SSPVIDP}
        J^{k+1}(s) =  \min_{\alpha \in Act} \bigg( c(s,\alpha) \\+ \sigma \left\{ J^k(s'),s,T(s'|s,\alpha) \right\} \bigg),~~\forall s \in\mathcal{S},
    \end{multline}
     converges to the optimal risk value function $J^*(s)$, $s \in \mathcal{S}$.\\
\textbf{(ii)} The optimal risk value functions $J^*(s)$, $s \in \mathcal{S}$ are the unique solution to the Bellman's equation
\begin{multline} \label{eq:SSPBellman}
            J^*(s) =  \min_{\alpha \in Act} \bigg( c(s,\alpha) \\+ \sigma \left\{ J^*(s'),s,T(s'|s,\alpha) \right\} \bigg),~~\forall s \in\mathcal{S};
\end{multline}
\textbf{(iii)} For any stationary Markovian policy $\mu$, the risk averse value functions $J(s,\mu(s))$, $s \in \mathcal{S}$ are the unique solutions to 
\begin{multline} \label{eq:SSPPI}
            J(s,\mu) =   c(s,\mu(s)) \\+ \sigma \left\{ J(s,\mu),s,T(s'|s,\alpha) \right\} ,~~\forall s \in\mathcal{S};
\end{multline}
\textbf{(iv)} A stationary Markovian policy $\mu$ is optimal if and only if  $\mu$ attains the minimum in Bellman's equation~\eqref{eq:SSPBellman}.
}
\end{theorem}

\begin{proof}
For every positive integer $M$, an initial state $s_0$, and policy $\pi$, we can split the nested risk cost~\eqref{cveerfdd}, where $\rho_{t,T}$ is defined in~\eqref{eq:dynriskmeasure},  at time index $\tau M$ and  obtain
\begin{multline} \label{dssddssds}
    J(s_0,\pi) = \rho_0 \big( c_0 + \cdots +
    \rho_{\tau M-1}(c_{\tau M-1}  \\ + \lim_{T\to\infty}
    \rho_{\tau M}\big(c_{\tau M} + \cdots+\rho_T(c_{T})  \cdots ))  \big),
\end{multline}
where we used the fact that the one-step coherent risk measures $\rho_t$ are continuous~\cite[Corollary 3.1]{ruszczynski2006optimization} and hence the limiting process and the measure $\rho_t$ commute. The limit term is indeed the total risk cost starting at $s_{\tau M}$, \textit{i.e.,} $J(s_{\tau M},\pi)$. 
Next, we show that under Assumption 1, this  term remains bounded. From~\eqref{eqsdsdsds} in the proof of Proposition 2, we have
\begin{align} \label{sdsssq2}
    |J(s_{\tau M}, & \pi)| = \lim_{T\to\infty}
    \rho_{\tau M}  \big(c_{\tau M} \cdots  +\rho_T(c_{T})\cdots) \nonumber \\& \le \sum_{k= M}^\infty
     \rho_{\tau k}\big(c_{\tau k}+\cdots+\rho_{\tau (k+1)-1}(c_{\tau (k+1)-1})\cdots\big) \nonumber \\
    &\le \sum_{k=M}^\infty \tau~\bar{c}~p^k = \frac{\tau \bar{c}~p^M}{1-p}.
    \end{align}
Substituting the above bound in~\eqref{dssddssds} gives
\begin{align}\label{sdce3}
    J(& s_0,\pi) \le \rho_0 \big(c_0 + \cdots +
    \rho_{\tau M-1}(c_{\tau M-1}  + \frac{\tau \bar{c}~p^M}{1-p})\cdots \big) \nonumber \\
    &= \rho_0 \big(c_0 + \cdots +
    \rho_{\tau M-1}(c_{\tau M-1})\cdots \big) + \frac{\tau \bar{c}~p^M}{1-p} ,
\end{align}
where the last equality holds via the translational invariance property of the one-step risk measures and the fact that $\frac{\tau \bar{c}~p^M}{1-p}$ is constant.
Similarly, following~\eqref{sdsssq2}, we can also obtain a lower bound on $J( s_0,\pi)$ as follows
\begin{align} \label{vcvddsaw2}
    J( s_0,\pi)  & \ge   \rho_0 \big(c_0 + \cdots +
    \rho_{\tau M-1}(c_{\tau M-1}  - \frac{\tau \bar{c}~p^M}{1-p})\cdots \big) \nonumber \\ &= \rho_0 \big(c_0 + \cdots +
    \rho_{\tau M-1}(c_{\tau M-1}) \big) - \frac{\tau \bar{c}~p^M}{1-p}.
\end{align}
Thus, from~\eqref{sdce3} and~\eqref{vcvddsaw2}, we obtain
\begin{align}\label{dsdsfsdc25}
      J(s_0,\pi)-\frac{\tau \bar{c}~p^M}{1-p} 
 &\le \rho_0 \big(c_0 + \cdots +
    \rho_{\tau M-1}(c_{\tau M-1}))\cdots \big) \nonumber \\
    &\le   J(s_0,\pi)+\frac{\tau \bar{c}~p^M}{1-p}. 
\end{align}
    
Furthermore, Assumption 1 implies that $J^0(s^g)=0$. If we consider $J^0$ as a terminal risk value function, we can obtain
\begin{subequations} \label{dfvd32}
\begin{align} 
    &|\rho_{\tau M}(J^0(s_{\tau M}))| = |\sup_{q \in \mathcal{Q}}~\langle J^0(s_{\tau M}), q \rangle_{\mathcal{Q}}| \\
   & = |\langle J^0(s_{\tau M}), q^* \rangle_{\mathcal{Q}}| \\
   &=  |\sum_{s \in \mathcal{S}} \mathbb{P}(s_{\tau M}=s  \mid s_0,\pi) q^*(s,\pi) J^0(s)| \\
    &\le \sum_{s \in \mathcal{S}} \mathbb{P}(s_{\tau M}=s  \mid s_0,\pi) q^*(s,\pi) \times \max_{s \in \mathcal{S}} |J^0(s)| \\ &\le
    \sum_{s \in \mathcal{S}} \mathbb{P}(s_{\tau M}=s  \mid s_0,\pi) \times \max_{s \in \mathcal{S}} |J^0(s)| \\
    & \le  p^M \max_{s \in \mathcal{S}} |J^0(s)|, 
\end{align}
\end{subequations}
where, similar to the derivation in~\eqref{eq:idk1}, in (18a) we used Proposition 1. Defining $q^* = \mathrm{argsup}_{q \in \mathcal{Q}} \langle \tau \bar{c}, q \rangle_{\mathcal{Q}}$, we obtained (18b) and the last inequality is based on the fact that the probability of $s_{\tau M} \neq s^g$ is less than equal to $p^M$ as in~\eqref{fcdfdscsdfc}. 


Combining inequalities~\eqref{dsdsfsdc25} and \eqref{dfvd32}, we have
\begin{align}
    - &p^M  \max_{s \in \mathcal{S}} |J^0(s)| + J(s_0,\pi)-\frac{\tau \bar{c}~p^M}{1-p} \nonumber \\
 &\le \rho_0 \big(c_0 + \cdots +
    \rho_{\tau M-1}(c_{\tau M-1} +\rho_{\tau M}(J^0(s_{\tau M}))\cdots \big) \nonumber \\
    &\le p^M \max_{s \in \mathcal{S}} |J^0(s)| + J(s_0,\pi)+\frac{\tau \bar{c}~p^M}{1-p}. \label{eq:mainineq1}
\end{align}
Remark that the middle term in the ineqaulity above is the $\tau M$-stage risk-averse cost of the policy $\pi$ with the terminal cost $J^0(s_{\tau M})$. Given Assumption 2, from~\cite[Theorem 2]{ruszczynski2010risk}, the minimum of this cost is generated by the dynamic programming recursion~\eqref{eq:SSPVIDP} after $\tau M$ iterations. Taking the minimum over the policy $\pi$ on every side of~\eqref{eq:mainineq1} yields
\begin{align}
    - &p^M  \max_{s \in \mathcal{S}} |J^0(s)| + J^*(s_0)-\frac{\tau \bar{c}~p^M}{1-p} \nonumber \\
   & \le J^{\tau M}(s_0)   \nonumber \\
    &\le p^M \max_{s \in \mathcal{S}} |J^0(s)| + J^*(s_0)+\frac{\tau \bar{c}~p^M}{1-p}, \label{eq:mainineq2}
\end{align}
for all $s_0$ and $M$. Finally, let $k=\tau M$. Since the above inequality holds for all $M$, taking the limit of $M \to \infty$ gives
\begin{equation} \label{eq322323}
\lim_{M \to \infty} J^{\tau M}(s_0)=\lim_{k \to \infty} J^{k}(s_0)=J^*(s_0), \quad \forall s_0 \in \mathcal{S}.
\end{equation}
\textbf{(ii)} Taking the limit, $k \to \infty$, of both sides of~\eqref{eq:SSPVIDP} yields
 $
        \lim_{k \to \infty}J^{k+1}(s) = \lim_{k \to \infty} \min_{\alpha \in Act} \big( c(s,\alpha) + \sigma \left\{ J^k(s'),s,T(s'|s,\alpha) \right\} \big),~~\forall s \in\mathcal{S}.
$ 
    Equality~\eqref{eq322323} in the proof of Part (i) implies that
        \begin{multline*}
        J^*(s) = \lim_{k \to \infty} \min_{\alpha \in Act} \bigg( c(s,\alpha) \\+ \sigma \left\{ J^k(s'),s,T(s'|s,\alpha) \right\} \bigg),~~\forall s \in\mathcal{S}.
    \end{multline*}
    Since the limit and the minimization commute over a finite number of alternatives, we have
        \begin{multline*}
        J^*(s) =  \min_{\alpha \in Act} \bigg( c(s,\alpha) \\+ \lim_{k \to \infty} \sigma \left\{ J^k(s'),s,T(s'|s,\alpha) \right\} \bigg),~~\forall s \in\mathcal{S}.
    \end{multline*}
    Finally, because $\sigma$ is continuous~\cite[Corollary 3.1]{ruszczynski2006optimization}, the limit and $\sigma$ commute as well and from~\eqref{eq322323}, we obtain 
$ 
        J^*(s) =  \min_{\alpha \in Act} \left( c(s,\alpha) \\+  \sigma \left\{ J^*(s'),s,T(s'|s,\alpha) \right\} \right),~~\forall s \in\mathcal{S}.
 $ 
    To show uniqueness, note that for any $J(s), s\in \mathcal{S}$ satisfying the above equation, the dynamic programming recursion~\eqref{eq:SSPVIDP} starting at $J(s), s\in \mathcal{S}$ replicates $J(s), s\in \mathcal{S}$ and from Part~(i) we infer $J(s)=J^*(s)$ for all $s \in \mathcal{S}$. \\\
\textbf{(iii)} Given a stationary Markovian policy $\mu$, at every state $s$, we have $\alpha = \mu(s)$, hence from Item (i), we have
    \begin{multline*}
        J^{k+1}(s,\mu) =  \min_{\alpha \in \{\mu(s)\}} \bigg( c(s,\alpha) \\+ \sigma \left\{ J^k(s',\mu),s,T(s'|s,\alpha) \right\} \bigg),~~\forall s \in\mathcal{S}.
    \end{multline*}
Since the minimum is only over  one element, we have
$
        J^{k+1}(s,\mu) =    c(s,\mu(s)) + \sigma \left\{ J^k(s',\mu),s,T(s'|s,\alpha) \right\} ,~~\forall s \in\mathcal{S}, 
$ 
    which with $k \to \infty$ converges uniquely (see Item (ii)) to $J(s,\mu)$.\\
\textbf{(iv)} The stationary policy attains its minimum in~\eqref{eq:SSPBellman}, if  
 \begin{multline*}
        J^*(s) =    \min_{\alpha \in \{\mu(s)\}} \bigg( c(s,\alpha) + \sigma \left\{ J^*(s'),s,T(s'|s,\alpha) \right\} \bigg) \\ = c(s,\mu(s)) + \sigma \left\{ J^*(s'),s,T(s'|s,\alpha) \right\} ,~~\forall s \in\mathcal{S}.
    \end{multline*}
Then, Part (iii) and the above equation imply $J(s,\mu) = J^*(s)$ for all $s \in \mathcal{S}$. Conversely, if $J(s,\mu) = J^*(s)$ for all $s \in \mathcal{S}$, then Items (ii) and (iii) imply that $\mu$ is optimal.
\end{proof}
\vspace{0.1cm}

At this point, we should highlight that, in the case of conditional expectation as the coherent risk measure ($\rho_t = \mathbb{E}$ and $\sigma \left\{ J(s'),s,T(s'|s,\alpha)\right\} = \sum_{s' \in \mathcal{S}}T(s'|s,\alpha)J(s')$), Theorem~1 simplifies to~\cite[Proposition 5.2.1]{dimitri1995dynamic} for the risk-neutral SSP problem. In fact, Theorem~1 is a generalization of~\cite[Proposition 5.2.1]{dimitri1995dynamic} to the risk-averse case. 

Recursion (dynamic programming)~\eqref{eq:SSPVIDP} represents the \emph{Value Iteration} (VI) for finding the risk value functions. In general, value iteration converges with infinite number of iterations ($k \to \infty$), but it can be shown, for a stationary policy $\mu$ resulting in an acyclic induced Markov chain, the VI algorithm converges in $|\mathcal{S}|$ of steps (see the derivation for the risk-neutral SSP problem in~\cite{Bertsekas99}). 

Alternatively, one can design risk-averse policies using \emph{Policy Iteration} (PI). That is, starting with an initial policy $\mu^0$, we can carry out \textit{policy evaluation} via~\eqref{eq:SSPPI} followed by a \textit{policy improvement} step, which calculates an improved policy $\mu^{k+1}$, as
 $
    \mu^{k+1}(s) = \argmin_{\alpha \in Act} \left( c(s, \alpha)  + \sigma \left\{ J^{\mu^k}(s,\alpha),s,T(s'|s,\alpha) \right\} \right),~~\forall s \in \mathcal{S}.
 $ 
This process is repeated until no further improvement is found in terms of the risk value functions: $J^{\mu^{k+1}}(s)=J^{\mu^k}(s)$ for all $s \in \mathcal{S}$. 



However, we do not pursue  VI or PI approaches further in this work. The main obstacle for using VI and PI is that equations \eqref{eq:SSPVIDP}-\eqref{eq:SSPPI} are nonlinear (and non-smooth) in the risk value functions for a general coherent risk measure. Solving nonlinear equations~\eqref{eq:SSPVIDP}-\eqref{eq:SSPPI} for the risk value functions may require significant computational burden (see the specialized non-smooth Newton Method in~\cite{ruszczynski2010risk} for solving similar nonlinear VIs). Instead, the next section present a computational method based on difference convex programs (DCPs).

\section{A DCP Computational Approach}

In this section, we propose a computational method based on DCPs to find the risk value functions and subsequently policies that minimize the accrued dynamic risk in the SSP planning. Before stating the DCP formulation, we show that the Bellman operator in~\eqref{eq:SSPVIDP} is non-decreasing. Let
\begin{align*}
\mathfrak{D}_\pi J : &= c(s,\pi(s)) +  \sigma\left( J(s'),s,T(s'|s,\pi(s)) \right), \quad \forall s \in \mathcal{S}, \\
\mathfrak{D} J : &= \min_{\alpha \in Act} \left( c(s,\alpha) +  \sigma\left( J(s'),s,T(s'|s,\alpha) \right) \right),~~\forall s \in \mathcal{S}.
\end{align*}
\begin{lemma}
\textit{Let Assumptions 1 and 2 hold. For all $v,w \in \mathcal{L}_m(\mathcal{S}, 2^\mathcal{S}, \mathbb{P})$, if $v \le w$, then $\mathfrak{D}_\pi v \le \mathfrak{D}_\pi w$ and $\mathfrak{D} v \le \mathfrak{D} w$.}
\end{lemma}
\vspace{0.1cm}
\begin{proof}
Since $\rho$ is a Markov risk measure, we have $\rho(v)=\sigma(v,s,T)$ for all $v \in \mathcal{L}_m(\mathcal{S}, 2^\mathcal{S}, \mathbb{P})$. Furthermore, since $\rho$ is a coherent risk measure from Proposition 1, we know that~\eqref{eq:dual} holds. Inner producting both sides of $v \le w$ with the probability measure $q \in \mathcal{Q}\subset \mathcal{P}$ from right and taking the supremum over $\mathcal{Q}$ yields 
$
\sup_{q \in \mathcal{Q}} ~~\langle v, q \rangle_{\mathcal{Q}} \le \sup_{q \in \mathcal{Q}} ~~\langle w, q \rangle_{\mathcal{Q}}.
$ 
From Proposition 1, we have $\sigma(v,s,T)=\sup_{q \in \mathcal{Q}} ~~\langle v, q \rangle_{\mathcal{Q}}$ and $\sigma(w,s,T)=\sup_{q \in \mathcal{Q}} ~~\langle w, q \rangle_{\mathcal{Q}}$. Therefore, 
 $
\sigma(v,s,T) \le \sigma(w,s,T).
$ 
Adding $c$ to both sides of the above inequality, gives $\mathfrak{D}_\pi v \le \mathfrak{D}_\pi w$. Taking the minimum with respect to $\alpha \in Act$ from both sides of $\mathfrak{D}_\pi v \le \mathfrak{D}_\pi w$, does not change the inequality and gives $\mathfrak{D} v \le \mathfrak{D} w$.
\end{proof}
\vspace{0.1cm}

We are now ready to state an optimization formulation to the Bellman equation~\eqref{eq:SSPVIDP}.
\vspace{0.1cm}
\begin{proposition}
\textit{Consider MDP~$\mathcal{M}$ as described in Definition~1. Let the Assumptions of Theorem 1 hold.  Then, the optimal value functions ${J}^*(s)$, $s\in \mathcal{S}$, are the solutions to the following optimization problem 
\begin{align}
  & \sup_{\boldsymbol{J}}~~\sum_{s\in \mathcal{S}}J(s) \label{eq:valueiteration}  \\
        &\text{subject to} \nonumber  \\
        &J(s) \le c(s,\alpha) + \sigma\{ {J}(s'),s,T(s'|s,\alpha) \},~\forall (s , \alpha)  \in \mathcal{S}\times {Act.}
        \nonumber
        \end{align}
}
\end{proposition}
\vspace{0.1cm}
\begin{proof}
From Lemma 1, we infer that $\mathfrak{D}_\pi$ and $\mathfrak{D}$ are non-decreasing; \textit{i.e.}, for $v\le w$, we have $\mathfrak{D}_\pi v \le \mathfrak{D}_\pi w$ and $\mathfrak{D} v \le \mathfrak{D} w$. Therefore, if $J \le \mathfrak{D} J$, then $\mathfrak{D} J \le \mathfrak{D}(\mathfrak{D} J)$. By repeated application of $\mathfrak{D}$, we obtain 
$
J \le \mathfrak{D} J \le \mathfrak{D}^2 J \le \mathfrak{D}^\infty J=J^*.
$ 
Any feasible solution to~\eqref{eq:valueiteration} must satisfy $J \le \mathfrak{D} J$ and hence must satisfy $J \le  J^*$. Thus, $J^*$ is the largest $J$ that satisfies the constraint in optimization~\eqref{eq:valueiteration}. Hence, the optimal solution to~\eqref{eq:valueiteration} is the same as that of~\eqref{eq:SSPVIDP}.  
\end{proof}
\vspace{0.1cm}

Once a solution $\boldsymbol{J}^*$ to optimization problem~\eqref{eq:valueiteration} is found, we can find a corresponding stationary Markovian policy as
\begin{multline*}
    \mu^*(s) = \argmin_{\alpha \in Act} \bigg( c(s, \alpha)  \\+ \sigma \left\{ J^*(s,\alpha),s,T(s'|s,\alpha) \right\} \bigg),~~\forall s \in \mathcal{S}.
\end{multline*}
\vskip -0.2 true in

\subsection{DCPs for Risk-Averse SSP Planning}
Assumption 1 implies that each $\rho$ is a coherent, Markov risk measure. Hence, the mapping $v \mapsto \sigma(v,\cdot,\cdot)$ is convex (because $\sigma$ is also a coherent risk measure). We next show that optimization problem~\eqref{eq:valueiteration} is in fact a DCP. 

Let $f_0=0$, $g_0(\boldsymbol{J})=\sum_{s\in \mathcal{S}} J(s)$, $f_1(\boldsymbol{J})=J(s)$, $g_1(s,\alpha)=c(s,\alpha)$, and $g_2(\boldsymbol{J})= \sigma(J,\cdot,\cdot)$. Note that $f_0$ and $g_1$ are convex (constant) functions and $g_0$, $f_1$, and $g_2$ are convex functions in $\boldsymbol{J}$. Then, ~\eqref{eq:valueiteration} can be expressed as the minimization
\begin{align}\label{eq:DCP}
      & \inf_{\boldsymbol{J}}~~ f_0-g_0(\boldsymbol{J})  \nonumber  \\
        &\text{subject to} \nonumber  \\
        &f_1(J)-g_1(s,\alpha)-g_2(J) \le 0, ~~\forall s,\alpha.
\end{align}

The above optimization problem is indeed a standard DCP~\cite{horst1999dc}. Many applications require solving DCPs, such as feature selection in  machine learning~\cite{le2008dc} and inverse covariance estimation in statistics~\cite{thai2014inverse}. DCPs can be solved globally~\cite{horst1999dc}, \textit{e.g.} using branch and bound algorithms~\cite{lawler1966branch}. Yet, a locally optimal solution can be obtained based on techniques of nonlinear optimization~\cite{Bertsekas99} more efficiently. In particular, in this work, we use a variant of the convex-concave procedure~\cite{lipp2016variations,shen2016disciplined}, wherein  the concave terms are replaced by a convex upper bound and solved. In fact, the disciplined convex-concave programming (DCCP)~\cite{shen2016disciplined} technique linearizes DCP problems into a (disciplined) convex program (carried out automatically via the DCCP Python package~\cite{shen2016disciplined}), which is then converted into an equivalent cone program by
replacing each function with its graph implementation. Then, the cone program can be solved readily by available convex programming solvers, such as CVXPY~\cite{diamond2016cvxpy}. 

In the Appendix, we present the specific DCPs required for risk-averse SSP planning for  CVaR and EVaR risk measures used in our numerical experiments in the next section. Note that for the risk-neutral conditional expectation measure, optimization~\eqref{eq:DCP} becomes a linear program, since~$\sigma \left\{ J(s'),s,T(s'|s,\alpha)\right\} = \sum_{s' \in \mathcal{S}}T(s'|s,\alpha)J(s')$ is linear in the decision variables $\boldsymbol{J}$.

\section{Numerical Experiments}\label{sec:example}

 \begin{figure}[t]\centering{
\includegraphics[scale=.27]{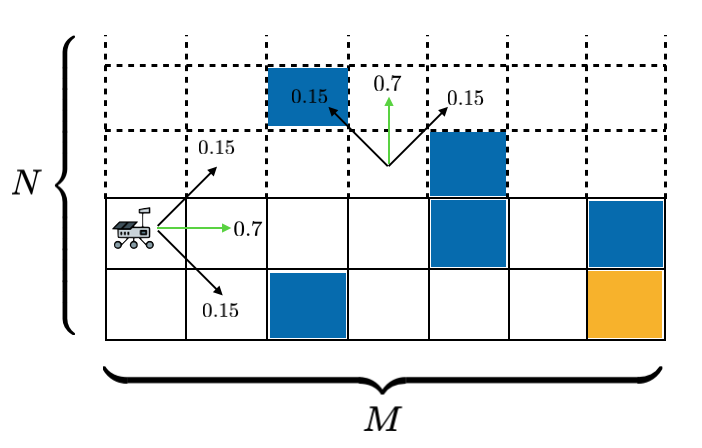}
\vspace{-0.4cm}
\caption{Grid world illustration for the rover navigation example. Blue cells denote the obstacles and the yellow cell denotes the goal.}
\vspace{-.5cm}
} 
\label{fig:gwa_again}
 \end{figure}
 
In this section, we evaluate the proposed method for risk-averse SSP planning with a rover navigation MDP (also used in~\cite{ahmadi2020risk,ahmadi2020constrained}). We consider the traditional total expectation as well as CVaR and EVaR. The experiments were carried out on a MacBook Pro with 2.8 GHz Quad-Core Intel Core i5 and 16 GB of RAM. The resultant linear programs and DCPs were solved using CVX~\cite{diamond2016cvxpy} with DCCP~\cite{shen2016disciplined} add-on.
 
An  agent (e.g. a rover) must autonomously navigate a 2-dimensional terrain map (e.g. Mars surface) represented by an $M \times N$ grid with $0.25 MN$ obstacles. Thus, the state space is given by
 $ \mathcal{S} = \{s^{i}|i=x+y,x\in\{1,\dots,M\},y \in \{1,\dots,N\}\}$ with $x=1,y=0$ being the leftmost bottom grid. 
Since the rover can move from cell to
cell, its action set is $ Act = \{E,\ W,\ N,\ S\}$. 
The actions move the robot from its current cell to a neighboring
cell, with some uncertainty. The state transition probabilities for various cell types are shown for actions $E$ (East) and $N$ (North) in Figure~2.  Other actions lead to similar transitions. Hitting an obstacle incurs the immediate cost  of $5$, while the goal grid region has zero immediate cost. Any other grid has a cost of $1$ to represent fuel consumption. 

Once the policies are calculated, as a robustness test similar to~\cite{chow2015risk,ahmadi2020risk,ahmadi2020constrained}, we included a set of single grid obstacles that are perturbed in a random direction to one of the neighboring grid
 cells with probability $0.2$ to represent uncertainty in the terrain map. For each risk measure, we run $100$ Monte Carlo simulations with the calculated policies and count the number of runs ending in a collision.
 
 In the experiments, we considered three grid-world sizes of $4\times 5$, $10 \times 10$, and $10 \times 20$ corresponding to $20$, $100$, and $200$ states, respectively. We allocated 2, 4, and 8 uncertain (single-cell) obstacles for the $4\times 5$, $10 \times 10$, and $10 \times 20$ grids, respectively.  In each case, we solve DCP~\eqref{eq:valueiteration} (linear program in the case of total expectation) with $|\mathcal{S}||Act|=MN \times 4 = 4MN$ constraints and $MN+1$ variables (the risk value functions $J$'s and $\zeta$ for CVaR and EVaR as discussed in the Appendix). In these experiments, we set the confidence levels to $\varepsilon= 0.3$ (more risk-averse) and $\varepsilon=0.7$ (less risk-averse) for both CVaR and EVaR coherent risk measures. The initial condition was chosen as $s_0=s^1$, \textit{i.e.,} the agent starts at the leftmost grid at the bottom, and the goal state was selected as $s^g = s^{MN}$, \textit{i.e.,}  the rightmost grid at the top.
 
 \begin{table}[t!]
\label{table:comparison}
\centering
\setlength\tabcolsep{2.5pt}
\begin{tabular}{lccccc} \midrule
\makecell{$(M \times N)_{\rho}$}  & \makecell{$J^*(s_0)$}  & \makecell{ Total \\  Time [s] }  & \makecell{\# U.O.}   & \makecell{ F.R.}   \\ 
\midrule
$(4\times 5)_{\mathbb{E}}$           & 13.25                          & 0.56 & 2                & 39\%                                        \\[1.5pt]
$(10\times 10)_{\mathbb{E}}$          & 27.31                          &  1.04  & 4           & 46\%                                             \\[1.5pt]
$(10\times 20)_{\mathbb{E}}$         & 38.35                          &  1.30 & 8               & 58\%                                         \\[1.5pt]
\midrule
$(4\times 5)_{\text{CVaR}_{0.7}}$            & 18.76                          &  0.58 & 2            &14\%                                            \\[1.5pt]
$(10\times 10)_{\text{CVaR}_{0.7}}$           & 35.72                          & 1.12   & 4           &19\%                                          \\[1.5pt]
$(10\times 20)_{\text{CVaR}_{0.7}}$         & 47.36                          & 1.36   & 8          &21\%                                           \\[1.5pt]
$(4\times 5)_{\text{CVaR}_{0.3}}$            & 25.69                          &  0.57 & 2            &10\%                                            \\[1.5pt]
$(10\times 10)_{\text{CVaR}_{0.3}}$           & 43.86                          & 1.16   & 4           &13\%                                          \\[1.5pt]
$(10\times 20)_{\text{CVaR}_{0.3}}$         & 49.03                          & 1.34   & 8          &15\%                                           \\[1.5pt]
\midrule
$(4\times 5)_{\text{EVaR}_{0.7}}$            & 26.67                         & 1.83  & 2          &9\%                                            \\[1.5pt]
$(10\times 10)_{\text{EVaR}_{0.7}}$          & 41.31                         & 2.02         &  4         &11\%                                      \\[1.5pt]
$(10\times 20)_{\text{EVaR}_{0.7}}$            & 50.79                         & 2.64      & 8           &17\%                                       \\[1.5pt]
$(4\times 5)_{\text{EVaR}_{0.3}}$            & 29.05                         & 1.79  & 2          &7\%                                            \\[1.5pt]
$(10\times 10)_{\text{EVaR}_{0.3}}$          & 48.73                         & 2.01         &  4         &10\%                                      \\[1.5pt]
$(10\times 20)_{\text{EVaR}_{0.3}}$            & 61.28                         & 2.78      & 8           &12\%                                       \\[1.5pt]
\midrule
\end{tabular}
\caption{Comparison between total expectation, CVaR, and EVaR  risk measures. $(M \times N)_{\rho}$ denotes the grid-world of size $M \times N$ and one-step coherent risk measure $\rho$. Total Time denotes the time taken by the CVX solver to solve the associated linear programs or DCPs. $\#$ U.O. denotes the number of single grid uncertain obstacles used for robustness test.  F.R. denotes the failure rate out of 100 Monte Carlo simulations. \vspace{-1cm}}
\end{table}

A summary of our numerical experiments is provided in Table~1. Note the computed values of Problem 1 satisfy $\mathbb{E}(c)\le \mathrm{CVaR}_\varepsilon(c) \le \mathrm{EVaR}_\varepsilon(c)$. This is in accordance with the theory that EVaR is a more conservative coherent risk measure than CVaR~\cite{ahmadi2012entropic} (see also our work on EVaR-based model predictive control for dynamically moving obstacles~\cite{dixit2020risksensitive}). Furthermore, the total accrued risk cost is higher for $\varepsilon =0.3$, since this leads to  more risk-averse policies. 

For total expectation coherent risk measure, the calculations took significantly less time, since they are the result of solving a set of linear programs. For CVaR and EVaR, a set of DCPs were solved. EVaR calculations were the most computationally involved, since they require solving exponential cone programs. Note that these calculations can be carried out offline for policy synthesis and then the policy can be applied for risk-averse robot path planning.

The table also outlines the failure ratios of each risk measure. In this case, EVaR outperformed both CVaR and total expectation in terms of robustness, which is consistent with the fact that EVaR is a more conservative risk measure. Lower failure/collision rates were observed for $\varepsilon = 0.3$, which correspond to more risk-averse policies. In addition, these results suggest that, although total expectation can be used as a measure of performance in high number of Monte Carlo simulations, it may not be practical to use it for real-world planning under uncertainty scenarios. CVaR and EVaR seem to be a more efficient metric for performance in shortest path planning under uncertainty.

\vspace{-.2cm}
\section{Conclusions} \label{sec:conclusions}

We proposed a method based on dynamic programming for designing risk-averse policies for the SSP problem. We presented a computational approach in terms of difference convex programs for finding the associated risk value functions and hence the risk-averse policies. Future research will extend to risk-averse MDPs with average costs and risk-averse MDPs with linear temporal logic specifications, where the former problem is cast as a special case of the risk-averse SSP problem. In this work, we assumed the states are fully observable, we will study the SSP problems with partial state observation~\cite{MSB19,ahmadi2020risk} in the future, as well. 


\footnotesize{
\bibliography{references}
}
\bibliographystyle{plain}

\normalsize

\section*{Appendix}


In this appendix, we present the specific DCPs for finding the risk value functions for two coherent risk measures studied in our numerical experiments, namely, CVaR and EVaR. 

For a given confidence level $\varepsilon \in (0,1)$, value-at-risk ($\mathrm{VaR}_\varepsilon$) denotes the $(1-\varepsilon)$-quantile value of the cost variable. $\mathrm{CVaR}_\varepsilon$ is the expected loss in the $(1-\varepsilon)$-tail given that the particular threshold $\mathrm{VaR}_\varepsilon$ has been crossed. $\mathrm{CVaR}_\varepsilon$ is given by 
\begin{equation}
    \rho_t(c_{t+1}) = \inf_{\zeta \in \mathbb{R}} \left\{ \zeta + \frac{1}{\varepsilon} \mathbb{E}\left[  (c_{t+1}-\zeta)_{+} \mid \mathcal{F}_t    \right]                        \right\},
\end{equation}
where $(\cdot)_{+}=\max\{\cdot, 0\}$. A value of $\varepsilon \simeq 1$ corresponds to a risk-neutral policy; whereas, a value of $\varepsilon \to 0$ is rather a risk-averse policy. 

In fact, Theorem 1 can applied to CVaR since it is a coherent risk measure. For MDP $\mathcal{M}$, the risk value functions can be computed by DCP~\eqref{eq:DCP}, where 
 $ 
g_2(J) = \inf_{\zeta\in \mathbb{R}} \left\{ \zeta + \frac{1}{\varepsilon} \sum_{s' \in \mathcal{S}}\left(J(s')-\zeta\right)_{+} T(s'\mid s,\alpha)        \right\},
$ 
where the infimum on the right hand side of the above equation can be absorbed into the overal infimum problem, \textit{i.e.,} $\inf_{\boldsymbol{J},\zeta}$. Note that $g_2(J)$ above is  convex in $\zeta$~\cite[Theorem 1]{rockafellar2000optimization}.

Unfortunately, CVaR ignores the losses below the VaR threshold (since it is only concerned with the average of VaR at the $(1-\epsilon)$-tail of the cost distribution). EVaR is the tightest upper bound in the sense of Chernoff inequality for the value at risk (VaR) and CVaR and its dual representation is associated with the relative entropy. In fact, it was shown in~\cite{ahmadi2017analytical} that $\mathrm{EVaR}_\varepsilon$ and $\mathrm{CVaR}_\varepsilon$ are equal only if there are no losses ($c\to 0$) below the $\mathrm{VaR}_\varepsilon$ threshold. In addition, EVaR is a strictly monotone risk measure; whereas, CVaR is only monotone~\cite{ahmadi2019portfolio}~(see Definition 4). $\mathrm{EVaR}_\varepsilon$  is given by
\begin{equation}
    \rho_t(c_{t+1}) = \inf_{\zeta >0} \left(  {\log \left(\frac{\mathbb{E}[e^{\zeta c_{t+1}} \mid \mathcal{F}_t]}{\varepsilon}\right)/ \zeta}        \right).            
\end{equation}
Similar to $\mathrm{CVaR}_\varepsilon$, for $\mathrm{EVaR}_\varepsilon$, $\varepsilon \to 1$ is a risk-neutral case; whereas, $\varepsilon\to 0$ corresponds to a risk-averse case. In fact, it was demonstrated in~\cite[Proposition 3.2]{ahmadi2012entropic} that $\lim_{\varepsilon\to 0} \mathrm{EVaR}_{\varepsilon}(c) = \esssup(c)=\bar{c}$ (worst-case cost). 

 Since $\mathrm{EVaR}_\varepsilon$ is a coherent risk measure, the conditions of Theorem 1 hold. Since $\zeta>0$, using the change of variables, $\tilde{\boldsymbol{J}} \equiv \zeta {\boldsymbol{J}}$ (note that this change of variables is monotone increasing in $\zeta$~\cite{agrawal2018rewriting}), we can compute EVaR value functions by solving~\eqref{eq:DCP}, where
 $
 f_0=0,
 f_1(\tilde{J})=\tilde{J}, 
 g_0(\tilde{{J}})=\sum_{s \in S} \tilde{J}(s),
 g_1(c)=\zeta c,$~{and}  
 $
 g_2(\tilde{{J}}) =   \log\Big(\frac{\sum_{s' \in \mathcal{S}}e^{ \tilde{J}(s' )}T(s'|s,\alpha)}{\varepsilon}\Big).
 $
 
Similar to the CVaR case, the infimum over $\zeta$ can be lumped into  the overall infimum problem, \textit{i.e.,} $\inf_{\tilde{\boldsymbol{J}},\zeta>0}$. Note that $g_2(\tilde{{J}})$ is  convex in  $\tilde{J}$, since the logarithm of sums of exponentials is convex~\cite[p.~72]{boyd2004convex}. 


\end{document}